\documentclass[11pt]{book}

\usepackage[dvips]{epsfig}

\usepackage{amssymb}
\usepackage{multirow}
\usepackage{amsbsy}
\usepackage{pstricks}

\newenvironment{proof}{\paragraph*{Proof:}}{\hfill$\square$}

\catcode`\@=11

\def\@normalsize{\@setsize\normalsize{13pt}\xipt\@xipt
	\abovedisplayskip 11pt plus3pt minus6pt
	\belowdisplayskip \abovedisplayskip
	\abovedisplayshortskip \z@ plus3pt
	\belowdisplayshortskip 6.6pt plus3.5pt minus3pt} 

\def\small{\@setsize\small{12pt}\xipt\@xipt
	\abovedisplayskip 10pt plus2pt minus5pt
	\belowdisplayskip \abovedisplayskip
	\abovedisplayshortskip \z@ plus3pt
	\belowdisplayshortskip 6pt plus3pt minus3pt
	\def\@listi{\topsep 6pt plus 2pt minus 2pt
		\parsep 3pt plus 2pt minus 1pt
		\itemsep \parsep}}

\def\footnotesize{\@setsize\footnotesize{10pt}\ixpt\@ixpt
	\abovedisplayskip 8pt plus 2pt minus 4pt
	\belowdisplayskip \abovedisplayskip
	\abovedisplayshortskip \z@ plus 1pt
	\belowdisplayshortskip 4pt plus 2pt minus 2pt
	\def\@listi{\topsep 4pt plus 2pt minus 2pt
		\parsep 2pt plus 1pt minus 1pt
		\itemsep \parsep}}

\def\scriptsize{\@setsize\scriptsize{9.5pt}\viiipt\@viiipt}
\def\tiny{\@setsize\tiny{7pt}\vipt\@vipt}
\def\large{\@setsize\large{14pt}\xiipt\@xiipt}
\def\Large{\@setsize\Large{18pt}\xivpt\@xivpt}
\def\LARGE{\@setsize\LARGE{22pt}\xviipt\@xviipt}
\def\huge{\@setsize\huge{25pt}\xxpt\@xxpt}
\def\Huge{\@setsize\Huge{30pt}\xxvpt\@xxvpt}

\def\section{\@startsection {section}{1}{\z@}%
	{-1.5\baselineskip plus-1pt minus-3pt}{1\baselineskip plus1pt minus2pt}%
	{\centering\normalsize\bf}}
\def\subsection{\@startsection{subsection}{2}{\z@}%
	{-1\baselineskip plus-1pt minus-2pt}{1\baselineskip plus1pt minus2pt}%
	{\normalsize\sc\noindent}}
\def\subsubsection{\@startsection{subsubsection}{3}{\z@}%
	{-1\baselineskip plus-1pt minus-2pt}{1sp}{\normalsize\it\noindent}}
\def\paragraph{\@startsection{paragraph}{4}{\z@}%
	{1\baselineskip plus1pt minus2pt}{-1em}{\normalsize\it\noindent}}
\let\subparagraph=\paragraph

\def\tableofcontents{\@restonecolfalse\if@twocolumn\@restonecoltrue
	\onecolumn\fi\OSIDcont\@starttoc{con}\if@restonecol\twocolumn\fi}

\def\l@section{\@dottedtocline{1}{0em}{.66em}}

\def\thebibliography#1{\section*{{Bibliography}\@mkboth
		{BIBLIOGRAPHY}{BIBLIOGRAPHY}}\footnotesize\rm\list
	{[\arabic{enumi}]}{\settowidth\labelwidth{[#1]}\leftmargin\labelwidth
		\advance\leftmargin\labelsep\usecounter{enumi}}
	\def\newblock{\hskip .11em plus .33em minus -.07em}
	\sloppy\clubpenalty4000\widowpenalty4000
	\sfcode`\.=1000\relax}

\def\ps@myheadings{\let\@mkboth\@gobbletwo
	\def\@oddhead{\hfil{\footnotesize\rm\rightmark}\hfil}
	\def\@evenhead{\hfil{\footnotesize\rm\leftmark}\hfil}
	\def\@oddfoot{\hfil{\footnotesize\sf\artid-\thepage}\hfil}
	\def\@evenfoot{\hfil{\footnotesize\sf\artid-\thepage}\hfil}
	\def\sectionmark##1{}\def\subsectionmark##1{}}

\newcounter{paPer}     %
\def\ps@osiD{\let\@mkboth\@gobbletwo
	\def\@oddfoot{\hfil{\footnotesize\sf\artid-\thepage}\hfil}
	\def\@evenhead{}\let\@evenfoot\@oddfoot}

\def\cite{\@ifnextchar [{\@tempswatrue\@Rcitex}{\@tempswafalse\@Rcitex[]}}

\def\@Rcitex[#1]#2{\if@filesw\immediate\write\@auxout{\string\citation{#2}}\fi
	\def\@citea{}\@cite{\@for\@citeb:=#2\do
		{\@citea\def\@citea{,\penalty\@m\,}\@ifundefined
			{b@\the\value{paPer}R\@citeb}{{\bf ?}\@warning
				{Citation `\@citeb' on page \thepage \space undefined}}%
			\hbox{\csname b@\the\value{paPer}R\@citeb\endcsname}}}{#1}}

\long\def\@caption#1[#2]#3{\par\addcontentsline{\csname
		ext@#1\endcsname}{#1}{\protect\numberline{\csname
			the#1\endcsname}{\ignorespaces #2}}\begingroup
	\@parboxrestore
	\small                                        
	\@makecaption{\csname fnum@#1\endcsname}{\ignorespaces #3}\par
	\endgroup}

\catcode`\@=12

\let\Rlabel=\label
\let\Rbibitem=\bibitem
\let\Rref=\ref
\let\Rpageref=\pageref
\def\label#1{\expandafter\Rlabel{\the\value{paPer}R#1}}
\def\bibitem#1{\expandafter\Rbibitem{\the\value{paPer}R#1}}
\def\ref#1{\expandafter\Rref{\the\value{paPer}R#1}}
\def\pageref#1{\expandafter\Rpageref{\the\value{paPer}R#1}}

\def\thesection{\arabic{section}.}

\def\YYMm{\rule{0ex}{4em}}
\newtoks\TITsi
\newtoks\TITsii

\def\title#1{\def\TITs{\LARGE{\raggedright\noindent\YYMm #1%
			\vskip8pt\par}}}

\def\author#1{\autMM{#1}\def\LHD{#1}}
\def\and{{\rm\lowercase{and}}}

\def\autMM#1{\TITsii={\vskip10pt\par\normalsize\rm\noindent #1\par}%
\TITsi=\expandafter{\TITs}\edef\TITs{\the\TITsi\the\TITsii}}

\def\address#1{\TITsii={\vskip6pt\par\footnotesize\sl\noindent #1\par}%
\TITsi=\expandafter{\TITs}%
\edef\TITs{\the\TITsi\the\TITsii}}

\def\received#1{\TITsii={\vskip10pt\par\small\rm\noindent(Received: #1)\par}%
\TITsi=\expandafter{\TITs}\edef\TITs{\the\TITsi\the\TITsii}}

\def\headtitle#1{\def\RHD{#1}}
\def\headauthor#1{\def\LHD{#1}}

\def\abst{{\bf Abstract.}}
\def\abstract#1{\TITs
\vskip15pt\par\noindent
{\footnotesize{\abst~} #1\vskip3pt\par}
\markright{\RHD}
\markboth{\LHD}{\RHD}}

\def\startpaper{%
\cleardoublepage
\setcounter{section}{0}
\stepcounter{paPer}
\setcounter{equation}{0}
\setcounter{footnote}{0}
\setcounter{figure}{0}
\setcounter{table}{0}
\def\theequation{\arabic{equation}}
\def\thefootnote{\arabic{footnote}}
\setcounter{defn}{0}
\setcounter{thm}{0}
\setcounter{lem}{0}
\setcounter{prop}{0}
\setcounter{rem}{0}
\thispagestyle{osiD}}

\def\OSIDcont{\cleardoublepage\thispagestyle{empty}
\markright{}\markboth{}{}
\normalsize\rm
\hspace*{\fill}{\large\rm
Contents of the Volume \Volume, Number \Number}\hspace*{\fill}
\par\vspace{1.5em}
\par\noindent}

\def\endpaper{\expandafter\label{\the\value{paPer}OpSy}}

\def\1{{\mathchoice{\rm 1\mskip-4mu l}{\rm 1\mskip-4mu l}%
{\rm 1\mskip-4.5mu l}{\rm 1\mskip-5mu l}}}

\def\varkappa{\mbox{\bBB\char 123}}

\def\longhookrightarrow{\lhook\joinrel\relbar\joinrel\rightarrow}

\def\longhookUp{\lower6pt\hbox{\rotatebox{90}{$\longhookrightarrow$}}}

\def\tr{\mathop{\rm tr}}

\newtheorem{thm}{\rm THEOREM}
\newtheorem{cor}{\rm COROLLARY}
\newtheorem{lem}{\rm LEMMA}

\newtheorem{prop}{\rm PROPOSITION}

\newtheorem{defn}{\rm DEFINITION}
\newtheorem{exmp}{\rm EXAMPLE}
\newtheorem{rem}{\it Remark}

\def\theequation{\thesection\arabic{equation}}

\def\Myskip{\setlength{\baselineskip}{13pt}}

\def\text#1{\quad\mbox{\rm  #1 }\quad}

\def\DOInumber{}

\input xy
\xyoption{all}

\InputIfFileExists{psfig.sty}{\typeout{^^Jpsfig.sty inputed...ok}}{\typeout{^^JWarning: psfig.sty could not be found.^^J}}

\usepackage{amsmath}
\usepackage{hyperref}

\begin{document}

\def\artid{0000001}
\def\Volume{15}
\def\Number{1}
\def\Year{2008}
\setcounter{page}{1}

\def\DOInumber{}

\startpaper

\newcommand{\Mn}{M_n(\mathbb{C})}
\newcommand{\Mk}{M_k(\mathbb{C})}
\newcommand{\id}{\mbox{id}}
\newcommand{\ot}{{\,\otimes\,}}
\newcommand{{\Cd}}{{\mathbb{C}^d}}
\newcommand{\sbsigma}{{\mbox{\scriptsize \boldmath $\sigma$}}}
\newcommand{\sbalpha}{{\mbox{\scriptsize \boldmath $\alpha$}}}
\newcommand{\sbbeta}{{\mbox{\scriptsize \boldmath $\beta$}}}
\newcommand{\bsigma}{{\mbox{\boldmath $\sigma$}}}
\newcommand{\balpha}{{\mbox{\boldmath $\alpha$}}}
\newcommand{\bbeta}{{\mbox{\boldmath $\beta$}}}
\newcommand{\bmu}{{\mbox{\boldmath $\mu$}}}
\newcommand{\bnu}{{\mbox{\boldmath $\nu$}}}
\newcommand{\ba}{{\mbox{\boldmath $a$}}}
\newcommand{\bb}{{\mbox{\boldmath $b$}}}
\newcommand{\sba}{{\mbox{\scriptsize \boldmath $a$}}}
\newcommand{\MD}{\mathfrak{D}}
\newcommand{\sbb}{{\mbox{\scriptsize \boldmath $b$}}}
\newcommand{\sbmu}{{\mbox{\scriptsize \boldmath $\mu$}}}
\newcommand{\sbnu}{{\mbox{\scriptsize \boldmath $\nu$}}}
\def\oper{{\mathchoice{\rm 1\mskip-4mu l}{\rm 1\mskip-4mu l}%
		{\rm 1\mskip-4.5mu l}{\rm 1\mskip-5mu l}}}
\def\<{\langle}
\def\>{\rangle}
\def\theequation{\thesection\arabic{equation}}

\thispagestyle{plain}

\title{Asymptotic dynamics in the Heisenberg picture: attractor subspace and Choi-Effros product}
\author{Daniele Amato$^{1,2}$, Paolo Facchi$^{1,2}$, Arturo Konderak$^{3,*}$}
\address{$^1$Dipartimento di Fisica, Universit\`a di Bari, I-70126 Bari, Italy} 
\address{$^2$INFN, Sezione di Bari, I-70126 Bari, Italy}
\address{$^3$Center for Theoretical Physics, Polish Academy of Sciences, 02-668 Warsaw, Poland}
\address{$^*$Corresponding author: \href{mailto:akonderak@cft.edu.pl}{\texttt{akonderak@cft.edu.pl}}}
\headauthor{D. Amato, P. Facchi, and A. Konderak}
\headtitle{Asymptotic dynamics in the Heisenberg picture}
\received{\today}

\abstract{We study the asymptotic dynamics of open quantum systems in the Heisenberg picture. We find an explicit expression for the attractor subspace and the dynamics that takes place in it. We present the relationship between the attractor subspaces in the Schr\"odinger and Heisenberg pictures and, in particular, the connection between their algebraic structures. An unfolding theorem of the asymptotics, as well as the fine structure of the recently introduced Choi-Effros decoherence-free algebra, are also discussed. Finally, we show how to extend all the results to the class of Schwarz maps.}

\Myskip


\section{Introduction}

Besides its fundamental interest, the study of the asymptotic evolution of open quantum systems has been stimulated over the last two decades~\cite{albert2019asymptotics,jex_st_2012,jex_st_2018,wolf2010inverse} by the impressive development of quantum technologies~\cite{koch2022quantum}. In particular, quantum information processing within decoherence-free and noiseless subspaces~\cite{zanardi1997noiseless,zanardi2014coherent},
as well as quantum reservoir engineering~\cite{kraus2008preparation,Wolf_res_eng}, require a solid understanding of the asymptotic subspace and the dynamics taking place in it. 

In view of the applications mentioned above, we need to go beyond the case of dynamics with a unique stationary state, already well studied in various seminal works~\cite{evans1977irreducible,Frigerio_78,Spohn_77} as well as in recent papers~\cite{fagnola2003transience,nigro2019uniqueness}. Most part of the works on the topic focused on the faithful case~\cite{carbone2020period,fagnola2019role,jex_st_2012,jex_st_2018}, which, as also revealed by the results of this Article, turns out to be much simpler than the general case. 
However, several recent papers started addressing the asymptotics of open quantum systems in a general setting. Specifically, Wolf~\cite{wolf2010inverse} and the authors~\cite{AFK_asympt,AFK_asympt_2} analyzed the problem in the Schr{\"o}dinger picture, whose structure is relevant for quantum communication tasks~\cite{Fawzi2025,singh2024zero,singh2024information}. On the other hand, Albert~\cite{albert2019asymptotics} considered the Heisenberg dynamics at large-times, also addressed by Bhat \textit{et al}~\cite{rajaramaperipheral_2,rajarama2022peripheral}, and the authors with an algebraic approach~\cite{AFK_asympt3,AFK_asympt_4}. Interestingly, universal constraints for the dimensions of the asymptotic subspaces of quantum evolutions were also found in~\cite{AF_bounds}.

The goal of the present work is to clarify further several aspects of the asymptotic dynamics of finite-dimensional open quantum systems. In particular, we will highlight the differences between faithful and non-faithful asymptotic evolutions by comparing known and original results. 

Specifically, the structure of the attractor subspace in the Schr\"odinger picture has been completely characterized in~\cite{AFK_asympt_2,wolf2010inverse}, together with the asymptotic dynamics in there, while a corresponding characterization is lacking when we deal with the Heisenberg picture. In this paper, we are going to exhibit such characterization and show that the corresponding structure is ultimately connected with the Choi-Effros product emerging in the asymptotic regime~\cite{rajarama2022peripheral}. 
%
It is worthwhile to observe that the mathematical structure of Heisenberg asymptotics is different to the one known in the Schr{\"o}dinger scheme. Such difference does not affect the expectation values, which can be experimentally measured, and thus the two pictures remain physically equivalent.
The asymptotic structure of Heisenberg dynamics is derived only with the help of Schwarz condition~\eqref{op_Schw_ineq}, so it remains valid within the broader class of Schwarz maps, cf.\ Section~\ref{Schw_CP}.\
Furthermore, we will manage to derive a decomposition for the Choi-Effros decoherence-free algebra, introduced in~\cite{AFK_asympt3} in order to study the multiplicative properties of the asymptotic dynamics. In Appendix~A we will also overview the structure of the peripheral eigenvectors of Heisenberg evolutions (Theorem~\ref{per_eigen_th}), already discussed in~\cite[Proposition 3]{albert2019asymptotics}.   

The Article is organized as follows. We will start in Section~\ref{sec:asymptotic_introduction} with an  introduction and a quick review of the main properties of the asymptotic dynamics of open quantum systems, both in the Schr\"odinger picture  (Subsection~\ref{sec:attr_schrodinger}) and in the Heisenberg picture (Subsection~\ref{intro_DFA_sec}). In Section~\ref{struc_attr_Heis_sec} we will present new results regarding the structure of the asymptotics in the Heisenberg picture (Theorem~\ref{struc_Heis_th}), and we will show that its algebraic properties are encoded in the ``faithful component'' of the evolution. Indeed, we show in Section~\ref{struc_DFA_sec} how the Choi-Effros product endowing the attractor subspace arises from the composition product equipping the attractor subspace of a (faithful) reduced evolution. We will also prove an unfolding theorem (Theorem~\ref{th:unfolding_thm}), which parallels the corresponding one in the Schr\"odinger picture~\cite[Theorem 4.1]{AFK_asympt_2}, and allows us to obtain any asymptotic dynamics from a  unital completely positive map. In Section~\ref{Choi-Effros_dec_sec} we will use the structure theorem~\ref{struc_Heis_th} to characterize the Choi-Effros decoherence-free algebra, showing its connection with the faithful part of the map. Importantly, in Section~\ref{Schw_CP} we will generalize our results to the much larger class of Schwarz maps. This shows that the asymptotic dynamics of an open quantum system is not determined by the full fire power of complete positivity, but rather by the much weaker Schwarz property~\eqref{op_Schw_ineq}.


\section{Asymptotic dynamics of open quantum systems}
\label{sec:asymptotic_introduction} 
\subsection{Preliminary notions}

In this section we will briefly recall the main properties of the asymptotic dynamics of open quantum systems.\ In particular, we will consider a finite-dimensional quantum system described by a $d$-dimensional Hilbert space $\mathcal{H}$. 
Let us first work in the Heisenberg picture, where we look at the dynamics of the observables of the system~\cite{Born_Jordan_1925,Born_Jordan_Heisenberg_1926,Heis_1925}. In such scheme the dynamics in the unit time is given by a unital completely positive (UCP) map $\Phi$ on $\mathcal{B}(\mathcal{H})$, see e.g. the recent review~\cite{chruscinski2022dynamical}. In the discrete-time setting of the dynamics, the evolved observable $A(n)$ at time $t=n\in \mathbb{N}$ will be the $n$-fold composition $\Phi^n$ of the map $\Phi$ on the observable $A$ at the initial time $t=0$, namely, $A(n)=\Phi^{n}(A)$.

The dynamics generated by $\Phi$ can be studied on the basis of its spectrum $\mathrm{spect}(\Phi)$, i.e.\ the set of its eigenvalues. It is possible to prove that all the eigenvalues have modulus smaller than or equal to 1, a property which is shared by the larger class of positive and unital maps~\cite{Asorey2008,watrous2018}. In fact, understanding the spectral features of completely positive maps remains essentially an open problem. See~\cite{chruscinski2021constraints,chruscinski2021universal} for interesting attempts in this direction. 

The asymptotic dynamics resulting from the $n\to\infty$ limit of the evolution takes place inside the asymptotic, peripheral or \emph{attractor subspace} of $\Phi$, defined as
\begin{equation}\label{attr_def}
	\mathrm{Attr}(\Phi):= \mbox{span} \{ X \in\mathcal{B}(\mathcal{H}) \,\vert\, \Phi(X)=\lambda X \mbox{ for some } \lambda \in  \mathrm{spect}_{\mathrm{P}}(\Phi)\} ,
\end{equation}
where $ \mathrm{spect}_{\mathrm{P}}(\Phi)$ denotes the \emph{peripheral spectrum} of $\Phi$:
\begin{equation}
	\mathrm{spect}_{\mathrm{P}}(\Phi):=\{ \lambda \in {\mathrm{spect}(\Phi)} \,\vert\, |\lambda|=1 \}.
\end{equation}
The eigenvectors corresponding to eigenvalues of modulus smaller than $1$ will vanish in the asymptotic limit, i.e. 
\begin{equation}
	\| \Phi^n(X) \| = \lvert 	\lambda \rvert^n \| X \| \rightarrow 0, \quad \lambda \in \mathrm{spect} (\Phi), \quad |\lambda|<1.
\end{equation}
Clearly, the attractor subspace $\mathrm{Attr}(\Phi)$ contains the \emph{fixed-point space} $\mathrm{Fix}(\Phi)$ of $\Phi$
\begin{equation}
	\label{fix_def} \mathrm{Fix}(\Phi) := \{ X \in \mathcal{B}(\mathcal{H}) \,\vert\, \Phi(X)=X \},
\end{equation}
namely the eigenspace of $\Phi$ corresponding to the eigenvalue $\lambda = 1$. Physically, the \emph{stationary observables} of the evolution given by $\Phi$ belong to the subspace $\mathrm{Fix}(\Phi)$.

The \emph{peripheral projection} $\mathcal{P}_{\mathrm{P}}$ of $\Phi$ is defined as
\begin{align}
	\label{P_p_def}
	\mathcal{P}_{\mathrm{P}}& :=\sum_{\lambda_k\in\mathrm{spect}_{\mathrm{P}}(\Phi)}\mathcal{P}_k,
\end{align}
where $\mathcal{P}_k$ denotes the spectral projection of $\Phi$ corresponding to the peripheral eigenvalue $\lambda_k \in \mathrm{spect}_{\mathrm{P}}(\Phi)$.
Clearly, the range of the projection $\mathcal{P}_{\mathrm{P}}$ is the attractor subspace $\mathrm{Attr}(\Phi)$. Let us call $\mathcal{P}$  the eigenprojection onto the fixed-point space $\mathrm{Fix}(\Phi)$, corresponding to $\lambda=1$.

Both projections $\mathcal{P}$ and $\mathcal{P}_{\mathrm{P}}$ can be written in terms of the UCP map $\Phi$ as~\cite{lindblad1999,wolf2012quantum}
\begin{equation}
	\mathcal{P}=\lim_{N\rightarrow \infty} \frac{1}{N}\sum_{n=1}^{N}\Phi^n, 
	\qquad
	\mathcal{P}_{\mathrm{P}}=\lim_{i\rightarrow \infty}\Phi^{n_i},\label{P_p_for}
\end{equation}
for some  subsequence $( n_i)_{i\in\mathbb{N}}$. 
From~\eqref{P_p_for} it is evident that both $\mathcal{P}$ and $\mathcal{P}_{\mathrm{P}}$ are UCP maps because $\Phi$ is, and the set of UCP maps forms a closed convex semigroup.

The space $\mathcal{B}(\mathcal{H})$ of bounded operators on $\mathcal{H}$ is a Banach algebra with respect to the adjoint $\ast$ and the operator norm $\| \cdot \|$, and in finite dimensions can be endowed with a Hilbert space structure by means of the \emph{Hilbert-Schmidt scalar product}
\begin{equation}
	\label{eq:hilbert_schmidt_inner_product}
	\left\langle{A}|{B}\right\rangle_{\mathrm{HS}} := \tr(A^\ast B), \quad A,B \in \mathcal{B}(\mathcal{H}).
\end{equation}

The Schr\"odinger dynamics, describing the dynamics of the states of the system, is given by $\Phi^\dagger$, the Hilbert-Schmidt adjoint of $\Phi$. As a result, $\Phi^\dagger$ is a completely positive trace-preserving (CPTP) map on $\mathcal{B}(\mathcal{H})$, and it is referred to as a \emph{quantum channel}~\cite{nielsen2002quantum}. Similarly to the Heisenberg picture, if $\rho$ is the initial state of the system, that is a positive semidefinite operator on $\mathcal{H}$ with unit trace,  $(\Phi^{\dagger})^n(\rho)$ describes the evolved state at time $t=n\in \mathbb{N}$. A state $\rho$ is said to be \emph{stationary} whenever it is invariant under the Schr\"odinger evolution, namely

\begin{equation}
	\Phi^\dagger(\rho)=\rho.
\end{equation}
Thus, we can define the attractor and fixed-point subspaces $\mathrm{Attr}(\Phi^\dagger)$ and $\mathrm{Fix}(\Phi^\dagger)$ of the Schr\"odinger dynamics $\Phi^\dagger$ as in equations~\eqref{attr_def} and~\eqref{fix_def} respectively. We can also introduce the peripheral and spectral projections onto $\mathrm{Attr}(\Phi^\dagger)$ and $\mathrm{Fix}(\Phi^\dagger)$  respectively which turn out to be the (Hilbert-Schmidt) adjoints $\mathcal{P}_{\mathrm{P}}^\dagger$ and $\mathcal{P}^\dagger$ of the projections $\mathcal{P}_{\mathrm{P}}$ and $\mathcal{P}$ of $\Phi$. Notice that the fixed-point subspace of a quantum channel is always non-empty~\cite[Theorem 4.24]{watrous2018}

An important family of UCP maps is given by the set of faithful maps. They will be the building block for the study of the asymptotic dynamics.
\begin{defn}
	\label{faith_def}
	A UCP map $\Phi$ is said to be \emph{faithful} if there exists an invertible stationary state $\rho$ for $\Phi^\dagger$, the Hilbert-Schmidt adjoint of $\Phi$, i.e.
	\begin{equation}
		\Phi^\dagger (\rho)=\rho >0.
	\end{equation} 
\end{defn}
%

We will see in the following sections that the assumption of a faithful map greatly simplifies the algebraic properties of the asymptotic dynamics.

\subsection{The attractor subspace in the Schr\"odinger picture}
\label{sec:attr_schrodinger}
In this section, we will provide a summary of the main results in the literature regarding the explicit structure of the attractor subspace  in the Schr\"odinger picture, whereas the corresponding issue in the Heisenberg picture will be addressed in Section~\ref{struc_attr_Heis_sec}.

As in the previous section, we will denote by $\Phi$ a UCP map, i.e.\ Heisenberg dynamics, and its corresponding Schr\"odinger dynamics is given by its Hilbert-Schmidt adjoint $\Phi^\dagger$. 

The asymptotic dynamics of $\Phi^\dagger$ is supported on the support of the fixed point $\mathcal P^\dagger(\mathbb I)$, namely given $X\in\mathrm{Attr}(\Phi^\dagger)$, its support $\mathrm{supp}(X) := \ker(X)^\perp$ and its range $\mathrm{ran}(X)$ satisfy
\begin{equation}
	\label{supp_fix}
	\mathrm{supp}(X),\mathrm{ran}(X) \subseteq \mathrm{supp}(\mathcal{P}^\dagger(\mathbb{I})):= \mathcal H_0 
\end{equation}
In particular, $\mathcal{P}^\dagger(\mathbb{I})$ is a maximum-rank fixed point of $\Phi^\dagger$. 
Importantly, as proved in~\cite[Appendix A]{AFK_asympt}, we also have that
\begin{equation}
	\mathcal{H}_{0} = \mathrm{supp}(\mathcal{P}_{\mathrm{P}}^\dagger (\mathbb{I})).
\end{equation}

The Hilbert space $\mathcal{H}$ can be decomposed as
\begin{equation}
	\mathcal{H}=\mathcal{H}_{0} \oplus \mathcal{H}_{1},
	\label{H_dec_supp}
\end{equation}
with $\mathcal{H}_1 = \mathcal{H}_{0}^\perp = \ker(\mathcal{P}_{\mathrm{P}}^\dagger(\mathbb{I}))$ the orthogonal complement of $\mathcal H_0$. Observe that equality $\mathcal{H}=\mathcal{H}_{0}$ is equivalent to say that the channel $\Phi^\dagger$ is faithful~\cite{AFK_asympt_2}.

We can now write the main result regarding the structure of the asymptotic dynamics in the Schr\"odinger picture~\cite{wolf2010inverse}.

\begin{thm}
	\label{Wolf_thm}
	Let $\Phi^\dagger: \mathcal B(\mathcal H)\rightarrow \mathcal B(\mathcal H)$ be a quantum channel. Then, there exists a decomposition of the Hilbert space $\mathcal H$ in the form
	\begin{equation}
		\label{eq:decomp_hilbert_space}
		\mathcal H= \left(\bigoplus_{k=1}^M \mathcal H_{k,1}\otimes \mathcal H_{k,2}\right)\oplus \mathcal H_1,
	\end{equation}
	where $\mathcal H_1= \ker (\mathcal P_{\mathrm{P}}^\dagger(\mathbb I))$, $\mathcal H_{k,i}$ are Hilbert spaces for all $k=1,\dots, M$ and $i=1,2$, and there are invertible density operators $\rho_k$ on $\mathcal H_{k,2}$ such that
	\begin{equation}
		\label{attr_struc_Schr}
		\mathrm{Attr}(\Phi^\dagger)= \left( \bigoplus_{k=1}^M \mathcal B(\mathcal H_{k,1})\otimes\mathbb C\rho_k\right)\oplus0.
	\end{equation}
	Thus, an element $X\in\mathrm{Attr}(\Phi^\dagger)$ has the form
	\begin{equation}
		X=\left(\bigoplus_{k=1}^M x_k\otimes \rho_k\right)\oplus 0 ,
	\end{equation}
	and the action of the quantum channel $\Phi^\dagger$ on $X$ is given by
	\begin{equation}
		\Phi^\dagger(X)=\left(\bigoplus_{k=1}^M U_kx_{\pi^{-1}(k)}U_k^\ast \otimes \rho_k\right)\oplus 0,
	\end{equation}	
	for some fixed permutation $\pi$ of $\{ 1, \dots , M \}$ with $d_{k}:= \dim(\mathcal{H}_{k,1}) = d_{\pi(k)}$ for all $k=1, \dots , M$, and unitary operators $U_k$ on $\mathcal{B}(\mathcal{H}_{k,1})$.
\end{thm}

The previous result gives an explicit expression for the asymptotics in the Schr\"odinger picture. Let us now clarify the general structure theorem~\ref{Wolf_thm} with the following simple example

\begin{exmp}\label{ex:amplit_damp}
    Consider the amplitude damping channel for a qubit, $\mathcal H=\mathbb C^2$, given by
    \begin{equation}\label{eq:amplit_damp}
        X=\begin{pmatrix}
            x_{00}&x_{01}\\
            x_{10}&x_{11}
        \end{pmatrix}\mapsto
        \Phi^\dagger(X)=\begin{pmatrix}
            x_{00}+\frac{3}{4}x_{11}&\frac{1}{2}x_{01}\\
            \frac{1}{2}x_{10}&\frac{1}{4}x_{11}
        \end{pmatrix}.
    \end{equation}
    Its peripheral projection reads
\begin{equation}
\label{per_proj_ampl_Schr}
\mathcal P_{\mathrm P}^{\dagger}(X)=
    \begin{pmatrix}
        x_{00}+x_{11}&0\\
        0&0
    \end{pmatrix}.
\end{equation}
Physically, the channel~\eqref{eq:amplit_damp} describes a probability flow from the orthogonal component $\mathcal H_0^{\perp}=\mathbb C |{1}\rangle$ to $\mathcal H_0=\mathbb C|{0}\rangle$, and this is a simple example of non-faithful channel. This could, for example, model the decay process of an excited atomic state.
\end{exmp}
More generally, we can always obtain a quantum channel $\Phi^\dagger$ with a desired asymptotic dynamics by using the unfolding theorem discussed in~\cite[Theorem 4.1]{AFK_asympt_2}, whose Heisenberg analogue is shown in Theorem~\ref{th:unfolding_thm} below.

We can now easily obtain insights regarding the Heisenberg dynamics in the faithful case. First of all, in the faithful case, the decomposition~\eqref{eq:decomp_hilbert_space} takes the form
\begin{equation}
	\label{eq:decomposition_hilbert_faithful}
	\mathcal H = \bigoplus_{k=1}^M \mathcal H_{k,1}\otimes \mathcal H_{k,2}.
\end{equation}
Moreover, the peripheral projection $\mathcal P_{\mathrm{P}}^\dagger$ in the Schr\"odinger picture is given by
\begin{equation}
	\label{eq:projection_faithful_case_shrodinger}
	\mathcal P_{\mathrm{P}}^\dagger(X)=\bigoplus_{k=1}^M \mathrm{tr}_{k,2}(P_kXP_k) \otimes \rho_k,
\end{equation}
where $P_k$ is the projection onto $\mathcal{H}_{k,1}\otimes \mathcal{H}_{k,2}$ and
$\mathrm{tr}_{k,2}$ denotes the partial trace  over $\mathcal{H}_{k,2}$. Its Hilbert-Schmidt adjoint, being the corresponding peripheral projection in the Heisenberg picture, reads  
\begin{equation}
	\label{eq:projection_faithful_case_heisenberg}
	\mathcal P_{\mathrm{P}}(X)=\bigoplus_{k=1}^M \mathrm{tr}_{k,2}\bigl(P_kXP_k(\mathbb I_{k,1}\otimes \rho_k)\bigr) \otimes{ \mathbb I_{k,2}},
\end{equation}
where the map
\begin{equation}
\label{partial_trace_Heis}
Y \in \mathcal{B}(\mathcal{H}_1 \otimes \mathcal{H}_2) \mapsto \mathrm{tr}_{2}(Y(\mathbb{I} \otimes \rho_2)) \in \mathcal{B}(\mathcal{H}_1) , \quad \rho_2 \in \mathcal{S}(\mathcal{H}_2),
\end{equation}
is the Hilbert-Schmidt adjoint of $\mathcal{B}(\mathcal{H}_{1}) \ni X \mapsto X \otimes \rho_2$, namely the addition of an uncorrelated ancillary state $\rho_2$~\cite[Chapter 9 Lemma 1.1]{davies1976}.\ Notice that the map~\eqref{partial_trace_Heis} is involved in Stinespring's dilation theorem of UCP maps~\cite{stinespring1955}.
Consequently, the attractor subspace of the Heisenberg evolution in the faithful case will be
\begin{equation}
	\label{eq:algebra_structure_attractor}
	\mathrm{Attr}(\Phi)=\mathrm{ran} (\mathcal{P}_{\mathrm{P}}) =
	\bigoplus_{k=1}^M \mathcal B(\mathcal H_{k,1})\otimes \mathbb{C} \mathbb I_{k,2}.
\end{equation}

Observe that  in the faithful case the attractor subspace is clearly a $C^*$-subalgebra of $\mathcal{B}(\mathcal{H})$, as it is closed under the standard composition product. We can also obtain the action of the asymptotic map in the Heisenberg picture under the faithfulness assumption. Given
\begin{equation}
	\label{struc_elem_attr}
	X=\bigoplus_{k=1}^M x_k \otimes \mathbb I_{k,2} \in\mathrm{Attr}(\Phi),
\end{equation}
we have
\begin{equation}
	\Phi(X)=\bigoplus_{k=1}^M U_{\pi(k)}^\ast x_{\pi(k)}U_{\pi(k)} \otimes\mathbb I_{k,2}.
	\label{eq:Phifaith}
\end{equation}
Indeed, given $X \in \mathrm{Attr}(\Phi)$ and $Y \in \mathcal{B}(\mathcal{H})$, then $\Phi(X) \in \mathrm{Attr}(\Phi)$ and consequently $\mathcal{P}_{\mathrm{P}}(\Phi(X)) = \Phi(X)$, so we have
\begin{equation}
	\tr(X\Phi^\dagger(Y)) = \tr(\Phi(X)Y) = \tr(\Phi(X)\mathcal{P}_{\mathrm{P}}^\dagger(Y)) := \tr(X\Phi^\dagger(Z)),
\end{equation}
where
\begin{equation}
	\mathrm{Attr}(\Phi^\dagger) \ni Z = \bigoplus_{k=1}^M z_k \otimes \rho_{k,2} \mapsto \Phi^\dagger(Z) = \bigoplus_{k=1}^M U_{k} z_{\pi^{-1}(k)} U_{k}^\ast\otimes \rho_{k,2} .
\end{equation}
by Theorem~\ref{Wolf_thm}. Therefore by using~\eqref{struc_elem_attr} 
\begin{equation}
	\begin{split}
		\tr(X\Phi^\dagger(Z)) &= \sum_{k=1}^M \tr( U_k^\ast x_k U_k z_{\pi^{-1}(k)} ) = \sum_{k=1}^M \tr(U_{\pi(k)}^\ast x_{\pi(k)} U_{\pi(k)} z_k) \\& := \tr(WZ) = \tr(WY), \; W := \bigoplus_{k=1}^M U_{\pi(k)}^\ast x_{\pi(k)} U_{\pi(k)} \otimes \mathbb{I}_{k,2},
	\end{split}
\end{equation}
where we used $\mathcal{P}_{\mathrm{P}}(W) = W$. Thus~\eqref{eq:Phifaith} readily follows by the arbitrariness of $Y$.
Moreover, we can easily check that $\Phi$ acts as a $*$-automorphism on $\mathrm{Attr}(\Phi)$ with respect to the composition product, in line with~\cite{rajarama2022peripheral}. As we shall see in the following section, this can be generalized to the non-faithful case by introducing an appropriate product. 

However, permutations may imply the lack of unitarity of the asymptotic dynamics, see~\cite{AFK_asympt_2} and~\cite[Example 3.2]{rajaramaperipheral_2}. Observe that, whenever we have a permutation $\pi$, the automorphism~\eqref{eq:Phifaith} is external, that is, it cannot be written as a unitary conjugation, with the unitary being an element of $\mathrm{Attr}(\Phi)$ itself. Explicitly, we cannot find a unitary $U\in\mathrm{Attr}(\Phi)$ such that
\begin{equation}
	\Phi(X) = U^\ast XU,\quad \forall X \in\mathrm{Attr}(\Phi).
\end{equation}


\subsection{Choi-Effros product}
\label{intro_DFA_sec}

We will conclude this introductory part by recalling the~\emph{Choi-Effros product}, which makes the attractor subspace of a UCP map a $C^\ast$-algebra in the non-faithful case, see~\cite[Theorem~2.3]{rajarama2022peripheral}. It was introduced by Choi and Effros in a seminal paper~\cite[Theorem~3.1]{choi1977injectivity} for injective operator systems.

\begin{thm}
	\label{th:C_star_UCP}
	Let $\Phi$ be a UCP map, and let $\mathcal P_{\mathrm{P}}$ be its peripheral projection. Then its attractor subspace $\mathrm{Attr}(\Phi)$ is a unital ${C}^\ast$-algebra with respect to the Choi-Effros product, defined as
	\begin{equation}
		\label{star_prod_H}
		X \star Y := \mathcal{P}_{\mathrm{P}}(XY),
	\end{equation}
	the involution and the norm being the native operations of $\mathcal{B}(\mathcal{H})$.
	
\end{thm} 

This theorem  makes the Choi-Effros product~\eqref{star_prod_H} the right candidate for describing the long-term behavior of an open quantum system in the Heisenberg scheme~\cite{AFK_asympt3}. Moreover, if $\Phi$ is faithful, then the Choi-Effros product $\star$ reduces to the usual composition product $\cdot$, and the algebraic structure of $\mathrm{Attr}(\Phi)$ is just given by equation~\eqref{eq:algebra_structure_attractor}. We will see in Section~\ref{struc_attr_Heis_sec} how to generalize such result to the non-faithful case.



In general, faithfulness is a sufficient but not necessary condition for the map $\Phi$ to be a $\ast$-morphism on the $C^\ast$-algebra $(\mathrm{Attr}(\Phi) , \cdot , \| \cdot \| , \ast)$. For this reason one introduces the class of peripherally automorphic maps~\cite{rajaramaperipheral_2}.

\begin{defn}
	Let $\Phi$ be a UCP map. Then $\Phi$ is called \emph{peripherally automorphic} iff
	\begin{equation}
		\label{per_aut_eq}
		X \star Y = XY, \quad \text{for all } X,Y \in \mathrm{Attr}(\Phi).
	\end{equation}
\end{defn} 	



Motivated by Theorem~\ref{th:C_star_UCP}, namely by the fact that the Choi-Effros product~\eqref{star_prod_H} is the suitable one for the study of Heisenberg asymptotics, we consider the space
\begin{align}
	\mathcal{N}_{\star}:= \{ X \in \mathcal{B}(\mathcal{H}) \,\vert\, \Phi^{n}(Y \star X) &= \Phi^{n}(Y)\star \Phi^{n}(X),\nonumber\\ \, \Phi^{n}(X \star Y) &= \Phi^{n}(X) \star \Phi^{n}(Y), \, \forall Y \in \mathcal{B}(\mathcal{H}),\forall n \in \mathbb{N} \},
	\label{star_dec_def}
\end{align}
which we refer to as \emph{Choi-Effros decoherence-free algebra of $\Phi$}~\cite{AFK_asympt3}.\ This generalizes the standard decoherence-free algebra, see e.g.~\cite{carbone2020period}. 

It turns out that $(\mathcal{N}_{\star} , \cdot , \| \cdot \| , \ast)$ is a $C^\ast$-algebra and it always contains the attractor subspace $\mathrm{Attr}(\Phi)$ of $\Phi$, at variance with the standard decoherence-free algebra~\cite{rajaramaperipheral_2}. A decomposition theorem for the Choi-Effros decoherence-free algebra will be recalled in Section~\ref{struc_DFA_sec} (Theorem~\ref{sum_dec_N_th}) and rewritten in terms of the decomposition~\eqref{H_dec_supp} of $\mathcal{H}$ (Theorem~\ref{dec_N_H01_th}). For further details on the Choi-Effros decoherence-free algebra see~\cite{AFK_asympt3}.

\section{Peripheral projection and attractor subspace}
\label{struc_attr_Heis_sec} 

In this section we will find the explicit structure of the attractor subspace $\mathrm{Attr}(\Phi)$ of the evolution in the Heisenberg picture, in analogy with the decomposition~\eqref{attr_struc_Schr} valid in the Schr\"odinger picture. For this purpose, we need to understand the structure of the peripheral projection $\mathcal{P}_{\mathrm{P}}$ of a generic UCP map $\Phi$.

First, it is useful to consider again the splitting~\eqref{H_dec_supp} of the Hilbert space:
\begin{equation}
	\label{eq:decomp_hilbert}
	\mathcal H=\mathcal H_0\oplus \mathcal H_1.
\end{equation}
Let $Q_0$ and $Q_1$ be the projections onto $\mathcal H_0$ and $\mathcal{H}_1$, respectively, and for any $X \in\mathcal B(\mathcal H)$ consider the operators $X_{ij}=Q_{i}X Q_{j}$. The operator $X$ can be written in a block-matrix form as
\begin{equation}
	X=
	\begin{pmatrix}
		X_{00} & X_{01}\\	
		X_{10} & X_{11}
	\end{pmatrix}.
\end{equation}

It is well known~\cite{jex_st_2012,wolf2012quantum} that it is possible to obtain a faithful channel by restricting $\Phi^\dagger$ to $\mathcal B(\mathcal H_0)$.  To be more precise, this subspace is invariant under $\Phi^\dagger$, and we can define 
${\varphi_{00}}:\mathcal{B}(\mathcal{H}_0) \mapsto \mathcal{B}(\mathcal{H}_{0})$ as
\begin{equation}
	X=
	\begin{pmatrix}
		X_{00} & 0\\	
		0 & 0
	\end{pmatrix}
	\mapsto
	\Phi^\dagger(X)=
	\begin{pmatrix}
		\varphi_{00}(X_{00}) & 0\\	
		0 & 0
	\end{pmatrix}
\end{equation}
It turns out that $\varphi_{00}$ is a faithful CPTP map on $\mathcal{B}(\mathcal{H}_0)$ and we will call it the \emph{reduced quantum channel} of $\Phi^\dagger$.
There exists a simple relationship between the fixed point spaces and the attractor subspaces of $\Phi^\dagger$ and $\varphi_{00}$, namely
\begin{equation}
	\label{Fix_Phi_Phi_tilde}
	\mathrm{Fix}(\Phi^\dagger) =
	\begin{pmatrix}
		\mathrm{Fix}(\varphi_{00}) & 0\\	
		0 & 0
	\end{pmatrix},\quad
	\mathrm{Attr}(\Phi^\dagger)=  
	\begin{pmatrix}
		\mathrm{Attr}(\varphi_{00}) & 0\\	
		0 & 0
	\end{pmatrix}.
\end{equation}

As a consequence of the structure of the Schr\"{o}dinger attractor subspace \eqref{attr_struc_Schr}, the peripheral projection $\mathcal{P}_{\mathrm{P}}^\dagger$ of 
$\Phi^\dagger$ 
reads
\begin{equation}
	\label{P_P_dagger_struc}
	\mathcal{P}_{\mathrm{P}}^\dagger(X)= 
	\begin{pmatrix}
		\sum_{i,j=0,1}\mathcal Q_{ij}(X_{ji})&0\\
		0 & 0
	\end{pmatrix}
\end{equation} 
Here, $\mathcal{Q}_{ij}: \mathcal{B}(\mathcal{H}_{i} , \mathcal{H}_{j}) \to \mathcal{B}(\mathcal{H}_0)$ are linear maps. In particular, $\mathcal Q_{00}$ is the peripheral projection of the faithful channel $\varphi_{00}$, and its explicit form is provided by~\eqref{eq:projection_faithful_case_shrodinger}.  In line with~\eqref{P_P_dagger_struc} and following the quantum probability jargon~\cite{gartner2012coherent}, we can call $\mathcal{B}(\mathcal{H}_0)$ and $\mathcal{B}(\mathcal{H}_{1})$ the \emph{recurrent} and \emph{transient algebras} of the evolution.

We can now consider the peripheral projection in the Heisenberg scheme $\mathcal{P}_{\mathrm{P}}$, whose action can be expressed as
\begin{equation}
	X=
	\begin{pmatrix}
		X_{00} & X_{01}\\
		X_{10} & X_{11}
	\end{pmatrix}\mapsto
	\mathcal{P}_{\mathrm{P}}(X)=\begin{pmatrix}
		\mathcal{Q}_{00}^\dagger(X_{00}) & \mathcal{Q}_{10}^\dagger (X_{00}) \\
		\mathcal{Q}_{01}^\dagger(X_{00}) & \mathcal{Q}_{11}^\dagger (X_{00})
	\end{pmatrix}.\label{eq:asympttic_map_heisenberg_block}
\end{equation}
Here, the adjoint maps $\mathcal{Q}_{ij}^\dagger : \mathcal{B}(\mathcal{H}_{0}) \to \mathcal{B}( \mathcal{H}_{i} , \mathcal{H}_{j})$ are again defined with respect to the Hilbert-Schmidt inner products~\eqref{eq:hilbert_schmidt_inner_product} on $\mathcal{B}(\mathcal{H}_{0})$ and $\mathcal{B}( \mathcal{H}_{i} , \mathcal{H}_{j})$. Henceforth the matrix components of $\mathcal{P}_{\mathrm{P}}$ will be denoted by $\mathcal{P}_{ij}=\mathcal{Q}_{ji}^\dagger$.

Notice that the action of the peripheral projection on $X$ only depends on $X_{00}=Q_0XQ_0$. Moreover, notice that $\mathcal P_{00}$ will be given by equation~\eqref{eq:projection_faithful_case_heisenberg}. Indeed,
it is the peripheral projection of the faithful \emph{reduced UCP map} $\phi_{00}=\varphi_{00}^\dagger:\mathcal B(\mathcal H_0)\rightarrow \mathcal B(\mathcal H_0)$, so its range, i.e.\ the attractor subspace of $\phi_{00}$, is the $C^\ast$-algebra

\begin{equation}
	\mathfrak A=\bigoplus_{k=1}^M \mathcal B(\mathcal H_{k,1})\otimes \mathbb{C}\mathbb I_{k,2}\subset \mathcal B(\mathcal H_0),
\end{equation}
and the action of $\mathcal P_{00}$ on $X_{00} \in \mathcal B(\mathcal H)$ is given by
\begin{equation}
	\label{eq:asympttic_map_heisenberg_00}
	\mathcal P_{00}(X_{00})=\bigoplus_{k=1}^M \mathrm{tr}_{k,2}\Bigl((\mathbb I_{k,1}\otimes \sqrt{\rho_k})P_kX_{00}P_k(\mathbb I_{k,1}\otimes \sqrt{\rho_k})\Bigr) \otimes{ \mathbb I_{k,2}}.
\end{equation}

	
	Now we are going to show that the expression for the asymptotic projection~\eqref{eq:asympttic_map_heisenberg_block} can be simplified. Remember that a UCP map satisfies the operator Schwarz inequality~\cite[Proposition 3.3]{paulsen2002}
	\begin{equation}
		\label{eq:op_Schw_ineq}
		\Phi(X^*X)\geqslant \Phi(X)^*\Phi(X), \quad \mbox{ for all } X \in \mathcal{B}(\mathcal{H}).
	\end{equation}
	We can exploit the latter in order to prove that $\mathcal{P}_{01}=\mathcal{P}_{10}=0$. Specifically, let us take $X=X_{00} \oplus 0$ with $X_{00}\in \mathfrak A$ and apply~\eqref{eq:op_Schw_ineq} to $\mathcal P_{\mathrm{P}}$ which is Schwarz since $\Phi$ is:
	\begin{equation}
		\label{eq:schrarz_periph}
		\mathcal{P}_{\mathrm{P}}(X_{00}^\ast X_{00} \oplus 0) \geqslant { \mathcal{P}_{\mathrm{P}}(X_{00}^\ast\oplus 0)\mathcal{P}_{\mathrm{P}}(X_{00}\oplus 0)=} \mathcal{P}_{\mathrm{P}}(X_{00}\oplus 0)^\ast\mathcal{P}_{\mathrm{P}}(X_{00}\oplus 0).
	\end{equation} 
	Using the block matrix expression~\eqref{eq:asympttic_map_heisenberg_block}, this reads
	\begin{align}
		&\begin{pmatrix}
			\mathcal{P}_{00}(X_{00}^*X_{00}) & \mathcal{P}_{01} (X_{00}^*X_{00}) \\
			\mathcal{P}_{10}(X_{00}^*X_{00}) & \mathcal{P}_{11} (X_{00}^*X_{00})
		\end{pmatrix}\geqslant\nonumber\\
		&\qquad\qquad
		\begin{pmatrix}
			\mathcal{P}_{00}(X_{00})^* & \mathcal{P}_{10} (X_{00})^* \\
			\mathcal{P}_{01}(X_{00})^* & \mathcal{P}_{11} (X_{00})^*
		\end{pmatrix}
		\begin{pmatrix}
			\mathcal{P}_{00}(X_{00}) & \mathcal{P}_{01} (X_{00}) \\
			\mathcal{P}_{10}(X_{00}) & \mathcal{P}_{11} (X_{00})
		\end{pmatrix}
	\end{align}
	From the upper-left block of the latter operator inequality we obtain
	\begin{equation}
		\mathcal{P}_{00}(X_{00}^\ast X_{00}) \geqslant \mathcal{P}_{00}(X_{00}^\ast)\mathcal{P}_{00}(X_{00}) + \mathcal{P}_{10}(X_{00})^\ast\mathcal{P}_{10}(X_{00}).
	\end{equation}
	Here, we used the fact that the operator $\mathcal P_{00}$ is Hermiticity preserving, so that $\mathcal P_{00}(X_{00}^\ast)=\mathcal P_{00}(X_{00})^\ast$. Moreover, since $\mathfrak A$ is a $\ast$-algebra, we have that $X_{00}^\ast X_{00} \in \mathfrak A$, namely
	\begin{equation}
		\mathcal P_{00}(X_{00})=X_{00}, \quad \mathcal P_{00}(X_{00}^\ast X_{00})=X_{00}^\ast X_{00}.
	\end{equation}
	As a result, we obtain
	\begin{equation}
		\label{eq:condition_off_channel}
		\mathcal{P}_{10}(X_{00})=0, \quad \forall X_{00} \in \mathfrak A.
	\end{equation}
	The idempotent condition $\mathcal{P}_{\mathrm{P}}^2=\mathcal{P}_{\mathrm{P}}$ implies
	\begin{equation}
		\label{idem_cond}
		\mathcal{P}_{ij}\mathcal{P}_{00}= \mathcal{P}_{ij}, \quad i,j=0,1,
	\end{equation}
	and thus $\mathcal{P}_{10}=0$. Analogously, by using the operator Schwarz inequality for $X^\ast_{00} \oplus 0$ and equation~\eqref{idem_cond}, we obtain $\mathcal{P}_{01}=0$. 
	
	To sum up, the peripheral projection $\mathcal{P}_{\mathrm{P}}$ takes the form
	\begin{equation}
		\label{eq:peripheral_projection}
		\mathcal{B}(\mathcal{H}) \ni X=
		\begin{pmatrix}
			X_{00} & X_{01}\\
			X_{10} & X_{11}
		\end{pmatrix}
		\mapsto
		\mathcal{P}_{\mathrm{P}}(X)=\begin{pmatrix}
			\mathcal{P}_{00}(X_{00}) & 0 \\
			0 & \mathcal{P}_{11} (X_{00})
		\end{pmatrix},
	\end{equation}
	with $\mathcal{P}_{00}$ being given by equation~\eqref{eq:asympttic_map_heisenberg_00} and, consequently, an idempotent UCP map. 
	Instead, $\mathcal{P}_{11}$ is required to be a UCP map with $\mathcal{P}_{11} \mathcal{P}_{00} = \mathcal{P}_{11}$, implying that $ \mathcal{P}_{11} \vert_{ \mathfrak A } $ is a generic UCP map. 
	Indeed, we will realize with the aid of the unfolding theorem~\ref{th:unfolding_thm} that there are no additional constraints on $\mathcal{P}_{11} \vert_{ \mathfrak{A} }$. See Remark~\ref{constraint_UCP_P_rmk}.
	
	Equation~\eqref{eq:peripheral_projection} was already obtained in~\cite[Theorem 7.1]{gartner2012coherent} for any normal idempotent UCP map on a von Neumann algebra $\mathcal{A}$, for which it is still possible to define recurrence and transience subalgebras as in the simple case $\mathcal{A}=\mathcal{B}(\mathcal{H})$. Note that the definition of faithfulness adopted in~\cite[Theorem 7.1]{gartner2012coherent}, provided just before the claim, is generally weaker than the one discussed in Definition~\ref{faith_def}. However, for idempotent positive unital maps such as $\mathcal{P}_{00}$, the two notions are equivalent; see the proof of Lemma~\ref{lemma_Schw_id_faith} and Remark~\ref{converse-faith-math-rem}. Also, we derived the Heisenberg analogue of equation~(22) of~\cite{albert2019asymptotics}. 
	
	As we shall see in Section~\ref{Schw_CP}, and as already clear in the proof, almost all features of this structure rely on the operator Schwarz inequality~\eqref{eq:op_Schw_ineq}.
	
	We summarize the results of the previous discussion in the following theorem.
	
	\begin{thm}
		\label{struc_Heis_th}
		Let $\Phi:\mathcal B(\mathcal H)\rightarrow \mathcal B(\mathcal H)$ be a UCP map. Then, it is possible to find a decomposition of the Hilbert space
		\begin{equation}
			\mathcal H=\left(\bigoplus_{k=1}^M \mathcal H_{k,1}\otimes \mathcal H_{k,2}\right)\oplus \mathcal H_1 ,
		\end{equation}
		with $\mathcal H_1=\ker ( \mathcal P_{\mathrm{P}}^\dagger(\mathbb I))$, and a UCP map
		\begin{equation}
			\label{eq:unital_map_attractor}
			\mathcal P_{11}:\mathfrak A 
			\rightarrow \mathcal B(\mathcal H_1), \qquad \text{with}\quad \mathfrak A =\bigoplus_{k=1}^M\mathcal B(\mathcal H_{k,1})\otimes \mathbb{C} \mathbb I_{k,2},
		\end{equation}
		so that the attractor subspace of $\Phi$ reads
		\begin{equation}
			\label{struc_attr_Heis}
			\mathrm{Attr}(\Phi)=\{X\oplus \mathcal P_{11}(X)|X\in\mathfrak A\}.
		\end{equation}
		The action of $\Phi$ on an element $X\oplus \mathcal{P}_{11}(X)$ of the attractor subspace reads
		\begin{equation}
			\label{as-struc_Heis}
			\Phi(X\oplus\mathcal P_{11}(X))=\phi_{00}(X) \oplus (\mathcal P_{11}\circ\phi_{00})(X).
		\end{equation}
		Here $X=\bigoplus_{k=1}^M x_k\otimes \mathbb I_{k,2}\in\mathfrak A$, and $\phi_{00}$ is defined as
		\begin{equation}
			\phi_{00}(X)=\bigoplus_{k=1}^M U_{\pi(k)}^\ast x_{\pi(k)}U_{\pi(k)} \otimes\mathbb I_{k,2},
		\end{equation}
		where $U_k$ are unitary operators on $\mathcal{B}(\mathcal{H}_{k,1})$, and $\pi$ is a permutation of the set $\{ 1, \dots , M \}$ with $d_k := \dim(\mathcal{H}_{k,1}) = d_{\pi(k)}$ for any $k=1, \dots , M$.
	\end{thm}
Note that Heisenberg's attractor subspace possesses a block diagonal structure with respect to the decomposition~\eqref{H_dec_supp}, exactly as the attractor subspace $\mathrm{Attr}(\Phi^\dagger)$ of the Schr\"odinger dynamics $\Phi^\dagger$. However, the non-faithfulness of $\Phi$ implies the appearance of an additional non-zero operator on $\mathcal{H}_{1}$, whereas in the Schr\"odinger analogue~\eqref{attr_struc_Schr} this operator is zero, implying an easier structure for $\mathrm{Attr}(\Phi^\dagger)$. Let us revisit Example~\ref{ex:amplit_damp} in the Heisenberg picture in order to illustrate Theorem~\ref{struc_Heis_th}.

\begin{exmp}
    The action of the amplitude damping channel~\eqref{eq:amplit_damp} in the Heisenberg picture is given by the UCP map
    \begin{equation}\label{eq:amplit_damp_1}
        X=\begin{pmatrix}
            x_{00}&x_{01}\\
            x_{10}&x_{11}
        \end{pmatrix}\mapsto
        \Phi(X)=\begin{pmatrix}
            x_{00}&\frac{1}{2}x_{01}\\
            \frac{1}{2}x_{10}&\frac{3}{4}x_{00}+\frac{1}{4}x_{11}
        \end{pmatrix},
    \end{equation}
with peripheral projection
\begin{equation}
\mathcal P_{\mathrm P}(X)=
    \begin{pmatrix}
        x_{00}&0\\
        0&x_{00}
    \end{pmatrix}.
\end{equation}
This example shows the difference between the Schr{\"o}dinger and Heisenberg asymptotic subspaces. In fact, the asymptotic observables have a non-zero residual component supported on the subspace $\mathbb{C}|{1}\rangle$, at variance with the asymptotic states, cf.\ equation~\eqref{per_proj_ampl_Schr}. 
\end{exmp}

	We conclude this section with the following result, which states that any dynamics of the form~\eqref{as-struc_Heis} can be obtained as the asymptotics of a proper UCP map. This result represents the analogue of the unfolding theorem discussed in the Schr\"odinger picture in~\cite[Theorem 4.1]{AFK_asympt_2}.
	
	\begin{thm}[Unfolding theorem]
		\label{th:unfolding_thm}
		Let $\mathcal H$ be a $d$-dimensional Hilbert space, and decompose it in the form
		\begin{equation}
			\mathcal H = \mathcal H_0\oplus \mathcal H_1,\qquad \mathcal H_0=\bigoplus_{k=1}^M \mathcal H_{k,1}\otimes \mathcal H_{k,2}.
		\end{equation}
		Call $d_k = \dim (\mathcal H_{k,1})$ the dimensions of $\mathcal H_{k,1}$. Let  $\mathcal P$ be a UCP map,
		\begin{equation}
			\mathcal P:\mathfrak A\rightarrow \mathcal B(\mathcal H_1), \qquad \text{with}\quad \mathfrak A=\bigoplus_{k=1}^M \mathcal B(\mathcal H_{k,1})\otimes \mathbb{C} \mathbb I_{k,2} \subset \mathcal B(\mathcal H_{0}).
		\end{equation}
		Consider the map
		\begin{equation}
			\Phi_{00}:\mathfrak A\rightarrow \mathfrak A,
		\end{equation}
		defined as
		\begin{equation}
			X=\bigoplus_{k=1}^M x_k \otimes \mathbb I_{k,2}\mapsto\Phi_{00}(X)=\bigoplus_{k=1}^M U^\ast_{\pi(k)}x_{\pi(k)}U_{\pi(k)}\otimes \mathbb I_{k,2},
			\label{eq:Phi00def}
		\end{equation}
		where $\pi$ be a permutation of $\{1,\dots,M\}$ such that $d_k=d_{\pi(k)}$, and   $U_k$ is a unitary in $\mathcal B(\mathcal H_{k,1})$ for all $k$.
		
		Consider the subspace of $\mathcal B(\mathcal H)$
		\begin{equation}
			\mathcal K=\{X\oplus \mathcal P(X)| X\in\mathfrak A\}.
		\end{equation}
		and let $\Phi_{\mathcal K}:\mathcal K\rightarrow \mathcal K$ be the map given by 
		\begin{equation}
			\Phi_{\mathcal K}(X\oplus \mathcal P(X))=\Phi_{00}(X)\oplus \mathcal P(\Phi_{00}(X)). 
		\end{equation}
		
		Then, there exists a UCP map $\Phi_{\mathrm E}$ on $\mathcal B(\mathcal H)$ such that
		\begin{itemize}
			\item[]$a.$ $\mathrm{Attr}(\Phi_{\mathrm E})=\mathcal K$
			\item[]$b.$ $\Phi_{\mathrm E}|_{\mathcal K}=\Phi_{\mathcal K}$
		\end{itemize}
	\end{thm}
	
	\begin{proof}
		Similarly to the proof of the unfolding theorem of the Schr\"odinger dynamics~\cite{AFK_asympt_2}, we will define the sought extension as the composition of three UCP maps. First, we define
		\begin{equation}
			\Phi_{\mathrm{pinch}}: \mathcal B(\mathcal H) \rightarrow \mathfrak A,
		\end{equation}
		
		so that
		
		\begin{equation}
			Z \mapsto \Phi_{\mathrm{pinch}}(Z)= \bigoplus_{k=1}^M \mathrm{tr}_{k,2}(P_k Z P_k) \otimes \frac{\mathbb I_{k,2}}{m_k} ,
			\label{proj_map}
		\end{equation}
		with $P_k$ the projection onto $\mathcal H_{k,1}\otimes \mathcal H_{k,2}$ and $m_k := \dim(\mathcal{H}_{k,2})$ for any $k=1, \dots , M$. The map~\eqref{proj_map} is UCP,  since it is of the form~\eqref{eq:projection_faithful_case_heisenberg} with $\rho_k = \mathbb{I}_{k,2} / m_k$. We take the second map to be $\Phi_{00}$ in~\eqref{eq:Phi00def}, which is  a $*$-automorphism over the algebra $\mathfrak A$, thus it is also UCP. Finally, we post-compose this map with
		\begin{align}
			\Lambda : \mathfrak A &\rightarrow \mathcal K,\\ 
			X&\mapsto X\oplus \mathcal P(X),
		\end{align}
		which is evidently UCP since $\mathcal{P}$ is.
		The composition of the three maps,
		\begin{equation}
			\Phi_{\mathrm E} = \Lambda\circ\Phi_{00}\circ{\Phi}_{\mathrm{pinch}},
		\end{equation}
		is the desired UCP map, satisfying conditions $a.$ and $b.$	
	\end{proof}
	
	\begin{rem}
		\label{constraint_UCP_P_rmk}
		Note that the UCP map $\mathcal P$ in Theorem~\ref{th:unfolding_thm} is completely general and does not satisfy further constraints. Therefore, the properties of the map $\mathcal{P}_{11} \vert_{\mathfrak A}$ appearing in equation~\eqref{struc_attr_Heis} are also sufficient for the description of the structure of the attractor subspace of a UCP map. 
	\end{rem}

	\section{Algebraic structure of the attractor subspace}
	\label{struc_DFA_sec} 

	In this section, we will explicitly show the relation between the structure of the attractor subspace of a UCP map described in the previous section and the algebraic structure of the reduced UCP map. In this way, we will justify the appearance of the Choi-Effros product characterizing the algebraic structure of the attractor subspace of a UCP map.

	To start with, consider the map
	\begin{align}
		\Lambda:\mathfrak A&\rightarrow \mathrm{Attr}(\Phi),\nonumber\\
		X_{00}&\mapsto \Lambda(X_{00})=
		\begin{pmatrix}
			X_{00} &  \\
			& \mathcal{P}_{11}(X_{00})
		\end{pmatrix}.
	\end{align}
	As observed in~\cite{gartner2012coherent}, $\Lambda$ defines a $*$-isomorphism between $\mathfrak A$ and $\mathrm{Attr}(\Phi)$, the latter equipped with the Choi-Effros product. Indeed, take $X_{00} , Y_{00} \in \mathfrak A$, and consider
	\begin{align}
		\Lambda(X_{00})\star\Lambda(Y_{00})
		&=\mathcal P_{\mathrm{P}}\left(
		\begin{pmatrix}
			X_{00} &  \\
			& \mathcal{P}_{11}(X_{00})
		\end{pmatrix}\begin{pmatrix}
			Y_{00} &  \\
			& \mathcal{P}_{11}(Y_{00})
		\end{pmatrix}\right)\nonumber\\
		& =  \mathcal P_{\mathrm{P}}
		\begin{pmatrix}
			X_{00}Y_{00} &  \\
			& \mathcal{P}_{11}(X_{00})\mathcal{P}_{11}(Y_{00})
		\end{pmatrix}\nonumber\\
		&=
		\begin{pmatrix}
			X_{00}Y_{00} &  \\
			& \mathcal{P}_{11}(X_{00}Y_{00})
		\end{pmatrix}
		=\Lambda(X_{00}Y_{00}).
	\end{align}
	
	Therefore $\Lambda^{-1}$ is a representation of the attractor subspace $\mathrm{Attr}(\Phi)$ of a UCP map $\Phi$ on  $\mathcal{B}(\mathcal{H}_0)$, the recurrence algebra of the evolution.
	Note that the UCP map $\Phi$ is peripherally automorphic if and only if the UCP map $\mathcal P_{11}$ is a $*$-morphism between $\mathfrak A$ and $\mathcal B(\mathcal H_1)$.

	The previous result shows the central role played by the map 
	\begin{equation}
		\Lambda: \mathfrak A \rightarrow \mathrm{Attr}(\Phi).
	\end{equation}
	Indeed, this map allows us to extend many known results in the faithful case to the non-faithful one. For instance, observe that, if $\Phi$ is faithful, then 
	\begin{equation}
    \label{corr_Heis_Schr}
		\mathrm{Attr} (\Phi) = \sigma^{-1/2}\mathrm{Attr}(\Phi^\dagger) \sigma^{-1/2},
	\end{equation}
	with $\sigma$ an invertible invariant state of the quantum channel $\Phi^\dagger$. The latter result relates the attractor subspaces in the Schr\"odinger and Heisenberg pictures of a faithful evolution~\cite[Subsection 4.3]{jex_st_2018}. In the general case, given a UCP map $\Phi$ with maximum-rank invariant state $\sigma$,~\eqref{corr_Heis_Schr} becomes    
	\begin{equation}
		\mathrm{Attr}(\Phi)=\Lambda(\sigma^{-1/2}\mathrm{Attr}(\Phi^\dagger) \sigma^{-1/2}),
	\end{equation}
	where we assumed $\sigma^{-1/2}$ to be $0$ on $\mathcal H_1$ and implicitly identified $\mathrm{Attr} (\phi_{00}) \oplus 0$ with $\mathrm{Attr} (\phi_{00})$.

	\section{Structure of the Choi-Effros decoherence-free algebra}
	\label{Choi-Effros_dec_sec}
	
	In the present section, we will exhibit an explicit structure for the Choi-Effros decoherence-free algebra $\mathcal{N}_{\star}$ of a UCP map, which is a useful tool we introduced in~\cite{AFK_asympt3} in order to understand the multiplicative properties of $\Phi$ in the asymptotic limit. The definition~\eqref{star_dec_def}, as well as some basic properties of this space, have been recalled in Subsection~\ref{intro_DFA_sec}. A main result about $\mathcal{N}_{\star}$ is the following decomposition theorem~\cite[Theorem 3.4]{AFK_asympt3}.

	\begin{thm}
		\label{sum_dec_N_th}
		Let $\Phi$ be a UCP map with peripheral projection $\mathcal{P}_{\mathrm{P}}$. Then its Choi-Effros decoherence-free algebra $\mathcal{N}_{\star}$ defined by equation~\eqref{star_dec_def} admits the following direct-sum decomposition
		\begin{equation}
			\label{sum_dec_N_star}
			\mathcal{N}_{\star}=\mathrm{Attr}(\Phi) \oplus \mathcal{K}_{\mathcal{P}_{\mathrm{P}}},
		\end{equation}
		where the $\ast$-ideal $\mathcal{K}_{\Psi}$ of any idempotent UCP map $\Psi$ is given by
		\begin{equation}
			\mathcal{K}_{\Psi}=\{ X \in \mathcal{B}(\mathcal{H}) \,\vert\, \Psi( X^\ast X)= \Psi( X X^\ast )=0  \}.
		\end{equation}
		Consequently $\mathcal{N}_{\star}$ is a ${C}^\ast$-algebra with respect to the composition product $\cdot$, the operator norm $\| \cdot \|$ and the adjoint $\ast$, or in short $(\mathcal{N}_{\star} , \cdot , \| \cdot \| , \ast)$ is a $C^\ast$-algebra. 
	\end{thm} 
	Our goal is to rewrite the decomposition~\eqref{sum_dec_N_star} in terms of the decomposition~\eqref{H_dec_supp} of the Hilbert space $\mathcal{H}$. By decomposing $X=(X_{ij})_{i,j=0} \in \mathcal{B}(\mathcal{H})$, it turns out that
	\begin{equation}
		\label{con_KPp}
		X \in \mathcal{K}_{\mathcal{P}_{\mathrm{P}}} \Leftrightarrow 
		\begin{cases}
			&\!\!\!\!\!\!\mathcal{P}_{00}(X_{00}^\ast X_{00} + X_{10}^\ast X_{10})=\mathcal{P}_{11}(X_{00}^\ast X_{00} + X_{10}^\ast X_{10})=0 \\
			&\!\!\!\!\!\!\mathcal{P}_{00}(X_{00} X_{00}^\ast + X_{01} X_{01}^\ast )=\mathcal{P}_{11}(X_{00} X_{00}^\ast  + X_{01} X_{01}^\ast )=0 
		\end{cases}.
	\end{equation}
	By the positivity of $\mathcal{P}_{ii}$, $i=0,1$, this is equivalent to
	\begin{align}
		& \mathrm{tr}( \mathcal{P}_{00}(X_{00}^\ast X_{00} + X_{10}^\ast X_{10}))=\mathrm{tr} (\mathcal{P}_{11}( X_{00}^\ast X_{00} + X_{10}^\ast X_{10}))=0 \label{con_KPp_tr} \\
		& \mathrm{tr} ( \mathcal{P}_{00}(X_{00} X_{00}^\ast + X_{01} X_{01}^\ast ) )=\mathrm{tr} (\mathcal{P}_{11} (X_{00} X_{00}^\ast + X_{01} X_{01}^\ast ) )=0.
		\label{con_KPp_tr_2}
	\end{align}
	Since $\mathcal{P}_{00}$ is faithful, $\mathcal{P}_{00}^{\dagger}(\mathbb{I}_{\mathcal{H}_0})>0$ (see equation~\eqref{supp_fix}), and equation~\eqref{con_KPp_tr} yields
	\begin{equation}
		\mathrm{tr}(\mathcal{P}_{00}(X_{00}^\ast X_{00} + X_{10}^\ast X_{10}))= \mathrm{tr}(\mathcal{P}_{00}^\dagger(\mathbb{I}_{\mathcal{H}_{0}})(X_{00}^\ast X_{00} + X_{10}^\ast X_{10}))=0,
	\end{equation}
	which implies $X_{00}=X_{10}=0$. Similarly, one can show from~\eqref{con_KPp_tr_2} that $X_{01}=0$. Therefore,
	\begin{equation}
		\mathcal{K}_{\mathcal{P}_{\mathrm{P}}}= 0 \oplus \mathcal{B}(\mathcal{H}_1),
	\end{equation}
	i.e. $\mathcal{K}_{\mathcal{P}_{\mathrm{P}}}$ is isomorphic to the transient algebra $\mathcal{B}(\mathcal{H}_1)$.
	This implies, together with the decomposition~\eqref{struc_attr_Heis}, that
	\begin{equation}
		\mathcal{N}_{\star}= \mathfrak A \oplus \mathcal{B}(\mathcal{H}_1),
	\end{equation}
	or, in a block matrix form,
	\begin{equation}
		\mathcal{N}_{\star}=
		\begin{pmatrix}
			\mathfrak A& \\
			& \mathcal B(\mathcal H_1)
		\end{pmatrix}.
	\end{equation}
	Therefore, we have the following novel decomposition for the algebra $\mathcal{N}_{\star}$.
	\begin{thm}
		\label{dec_N_H01_th}
		Let $\Phi$ be a UCP map with peripheral projection $\mathcal{P}_{\mathrm{P}}$ and Choi-Effros decoherence-free algebra $\mathcal{N}_{\star}$. Let $\mathcal{H}_{0}=\mathrm{supp}(\mathcal{P}_{\mathrm{P}}^\dagger(\mathbb{I}))$. Then the algebra $\mathcal{N}_{\star}$ can be decomposed as
		\begin{equation}
			\label{dec_N_H01}
			\mathcal{N}_{\star} = \mathfrak A \oplus \mathcal{B}(\mathcal{H}_{1}),
		\end{equation}
		where $\mathcal{H}_1 = \mathcal{H}_0^\perp$ and $\mathfrak A$ is the attractor subspace of the reduced UCP map $\phi_{00}$ of $\Phi$.
	\end{thm}
	If $\Phi$ is faithful, then $\mathcal{H}_0^\perp = 0$, and consequently $\mathrm{Attr}(\Phi)=\mathcal{N}_{\star}$, as already seen in~\cite[Theorem 3.11]{AFK_asympt3}. In fact, according to the latter result, this condition characterizes faithful evolutions $\Phi$.
	

	\section{Some considerations on the Schwarz inequality}
	\label{Schw_CP} 
	In the previous sections, we consistently discussed the Heisenberg open-system dynamics characterized by a UCP map $\Phi$. In the current section, we will extend the previous findings to the broader class of Schwarz maps. A positive unital map is said to be a \emph{Schwarz map}~\cite{wolf2010inverse} if it satisfies 
	\begin{equation}
		\label{op_Schw_ineq}
		\Phi(X^\ast X) \geqslant \Phi(X)^\ast \Phi(X),  \quad \mbox{ for all } X \in \mathcal{B}(\mathcal{H}), 
	\end{equation}
	called \emph{operator Schwarz inequality}. All UCP maps satisfy equation~\eqref{op_Schw_ineq}, while there are Schwarz maps that are not UCP~\cite{choi1980some,chruscinski2020kadison}.

	Let $\Phi$ be a Schwarz map, and note that we can define the attractor subspace $\mathrm{Attr}(\Phi)$ as well as the peripheral projection $\mathcal{P}_{\mathrm{P}}$ of $\Phi$ as in equations~\eqref{attr_def} and~\eqref{P_p_def}. In particular, as a consequence of~\eqref{P_p_for}, the peripheral projection $\mathcal{P}_{\mathrm{P}}$ of $\Phi$ is a Schwarz map since $\Phi$ is.
	
	Let us generalize Theorem~\ref{struc_Heis_th}  to Schwarz maps. First, we can define once more the reduced map of $\Phi$, 
	\begin{equation}
		\phi_{00}:\mathcal B(\mathcal H_0)\rightarrow \mathcal B(\mathcal H_0),
	\end{equation}
	which turns out to be a faithful Schwarz map on $\mathcal{B}(\mathcal{H}_{0})$. Therefore, its peripheral projection $\mathcal{P}_{00}$ is a Schwarz map too, and it is also faithful. Let us now prove that a faithful idempotent Schwarz map is automatically completely positive. 
	\begin{lem}
		\label{lemma_Schw_id_faith}
		Let $\Phi$ be a faithful idempotent Schwarz map, i.e. $\Phi^2 = \Phi$. Then it is completely positive.
	\end{lem}
	\begin{proof}
		Given a positive unital map $\Phi$, faithfulness is defined as $\Phi^\dagger(\mathbb{I})>0$. Therefore, given $A \geqslant 0$, we have
		\begin{equation}
			\label{faith_math}
			\Phi(A)=0 \Leftrightarrow \mathrm{tr}(\Phi^\dagger(\mathbb{I})A)=\mathrm{tr}(\Phi(A))=0 \Leftrightarrow A=0.
		\end{equation}
		The assertion now follows from~\cite[Corollary 2.2.7]{stormer2013positive} thanks to the idempotence of $\Phi$.
	\end{proof}
	
	\begin{rem}
		\label{converse-faith-math-rem}
		Conversely, equation~\eqref{faith_math} immediately implies the faithfulness of $\Phi$. 
	\end{rem}
	
	As a consequence, $\mathcal{P}_{00}$ is a UCP map even when ${\Phi}_{00}$ is only a Schwarz map. It is possible to repeat the proof of Theorem~\ref{struc_Heis_th} for Schwarz maps $\Phi$. The only difference concerns $\mathcal{P}_{11}$ involved in the action~\eqref{eq:unital_map_attractor} of the projection $\mathcal{P}_{\mathrm{P}}$. 
	Specifically, the following result holds.
	
	\begin{thm}
		\label{struc_Heis_th_Schw}
		Let $\Phi$ be a Schwarz map. Its peripheral projection $\mathcal{P}_{\mathrm{P}}$ takes the form
		\begin{equation}
			\label{struc_P_p_Heis_Schw}
			\mathcal{P}_{\mathrm{P}}(X)=\mathcal{P}_{00}(X_{00}) \oplus \mathcal{P}_{11}(X_{00}),
		\end{equation}
		where the UCP map $\mathcal{P}_{00} $ is the peripheral projection of the reduced map $\phi_{00}$ of $\Phi$ defined by~\eqref{eq:unital_map_attractor} and $\mathcal{P}_{11} : \mathrm{Attr}(\phi_{00}) \rightarrow \mathcal{B}(\mathcal{H}_{1})$ is a Schwarz map. An element $X \in \mathrm{Attr}(\Phi)$, the attractor subspace of $\Phi$, admits the following decomposition
		\begin{equation}
			X = X_{00} \oplus \mathcal{P}_{11}(X_{00}), \quad X_{00}\in \mathrm{Attr}(\phi_{00}).
		\end{equation}
	\end{thm}
	
	The only difference from Theorem~\ref{struc_Heis_th} is that $\mathcal{P}_{11}$ is \emph{not} generally completely positive. Indeed, if $\mathcal{P}_{11}$ were a UCP map, then it is clear that $\mathcal{P}_{\mathrm{P}}$ would be UCP too. But we know that it is possible to construct Schwarz idempotent maps that are not  completely positive when $d \geqslant 4$, see~\cite[Example pg. 77, Theorem 3.1]{osaka1991positive}. 
	
	We can prove that $\mathcal{P}_{11}$ and $\mathcal{P}_{\mathrm{P}}$ can fail to be completely positive only if the Schwarz map $\Phi$ is not peripherally automorphic. In fact, Lemma~\ref{lemma_Schw_id_faith} can be generalized to peripherally automorphic maps.

	\begin{prop}
		\label{prop_idem_pos_per_aut}
		Let $\Phi$ be a peripherally automorphic unital positive idempotent map. Then it is completely positive. 
	\end{prop}
	\begin{proof}
		Since $\Phi$ is peripherally automorphic, the range $\mathrm{Ran}(\Phi)=\mathrm{Attr}(\Phi)$ of $\Phi$ is a $C^\ast$-subalgebra of $\mathcal{B}(\mathcal{H})$, and the statement follows from~\cite[Theorem 1.5.10]{brown2008}.
	\end{proof}
	\begin{rem}
	The proof strategy of Lemma~\ref{lemma_Schw_id_faith} cannot be applied in order to prove Proposition~\ref{prop_idem_pos_per_aut}, even if we assume $\Phi$ to be a Schwarz map. In fact~\eqref{faith_math} is not generally true for peripherally automorphic idempotent Schwarz maps, see~\cite[Example 2.11 (1)]{rajaramaperipheral_2}.	 
	\end{rem}
	\begin{cor}
		\label{prop_Schw_id_per_aut}
		Let $\Phi$ be a peripherally automorphic Schwarz map. Then, following the notation of Theorem~\ref{struc_Heis_th_Schw}, $\mathcal{P}_{\mathrm{P}}$, and consequently $\mathcal{P}_{11}$, are completely positive.
	\end{cor}
	\begin{proof}
		$\mathcal{P}_{\mathrm{P}}$ is a peripherally automorphic Schwarz idempotent map by the properties of $\Phi$, and the corollary immediately follows from Proposition~\ref{prop_idem_pos_per_aut}.
	\end{proof}
	
	Roughly speaking, replacing the complete positivity property with the operator Schwarz inequality for the map $\Phi$ simply gives more freedom to the operator $\mathcal{P}_{11}$. Importantly, the attractor subspace $\mathrm{Attr}(\Phi)$ of $\Phi$ still admits a block diagonal decomposition in terms of the decomposition~\eqref{H_dec_supp} of $\mathcal{H}$. Observe that the decomposition~\eqref{dec_N_H01} for the algebra $\mathcal{N}_{\star}$ discussed in Section~\ref{struc_DFA_sec} still holds for Schwarz maps, as well as the initial decomposition~\eqref{sum_dec_N_star} recalled in the same section.

	\section{Conclusions} 

	In this work we discussed open-system asymptotic dynamics in the Heisenberg picture. First, we explicitly found the structure of the attractor subspace, and we characterized the evolution on this subspace. In particular, we found that the faithful component of the attractor subspace is mapped into a (transient) orthogonal component by an arbitrary UCP map. On the level of the evolution and as in the Schr\"odinger picture, permutations can make the asymptotic dynamics non-unitary in general,  even in the faithful case. Mathematically, permutations are described by an external $*$-morphism. Additionally, superselection rules appear between subspaces.
	
	In this regard, it would be interesting to derive a complete characterization of the unitarity of the asymptotic map for Heisenberg dynamics like the one considered in~\cite{AFK_asympt} for Schr\"odinger evolutions.	

Subsequently, we showed how the Choi-Effros product, endowing the attractor subspace of a Heisenberg evolution,  emerges from the composition product equipping the attractor subspace of the reduced dynamics. Interestingly, the latter is a faithful representation of the attractor subspace of the Heisenberg dynamics on a smaller Hilbert space~\cite{gelfand1943imbedding}. Also, we obtained an explicit expression for the Choi-Effros decoherence-free algebra~\eqref{star_dec_def}, that was introduced in~\cite{AFK_asympt3} as a natural generalization of the standard decoherence-free algebra.

In conclusion, we showed that all the results obtained in the present work are valid for the larger class of Schwarz maps, so the operator Schwarz inequality seems to be the crucial property in the study of the asymptotics of open quantum systems. 

As a final remark, all the results discussed in this work may be generalized to finite-dimensional $C^\ast$-algebras in a straightforward way.

	\appendix
	\section*{Appendix A: Peripheral eigenvectors of $\Phi$}
	\def\theequation{A.\arabic{equation}}
	\setcounter{equation}{0}
	\label{albert_proof} 

	In this Appendix we will  overview 
	the structure of the peripheral eigenoperators of a generic UCP map $\Phi$.
	%
	To this purpose, we will mimic the path followed in Section~\ref{struc_attr_Heis_sec} for the derivation of Theorem~\ref{struc_Heis_th}. We will consider the action of the map $\Phi$ on block diagonal elements. Specifically,	\begin{equation}
		X=\begin{pmatrix}
			X_{00}& 0\\
			0 & 0
		\end{pmatrix}\mapsto
		\Phi(X)=\begin{pmatrix}
			\phi_{00}(X_{00})& \phi_{01}(X_{00})\\
			\phi_{10}(X_{00}) & \phi_{11}(X_{00})
		\end{pmatrix},
	\end{equation}
	and
	\begin{equation}
		Y=\begin{pmatrix}
			0& 0\\
			0 & Y_{11}
		\end{pmatrix}\mapsto
		\Phi(Y)=\begin{pmatrix}
			\psi_{00}(Y_{11})& \psi_{01}(Y_{11})\\
			\psi_{10}(Y_{11}) & \psi_{11}(Y_{11})
		\end{pmatrix},
	\end{equation}
	with $X_{00}$ in $\mathcal B(\mathcal H_{00})$, $Y_{11}$ in $\mathcal B(\mathcal H_{11})$, and $\phi_{ij}: \mathcal{B}(\mathcal{H}_{0})\mapsto \mathcal{B}(\mathcal{H}_{j} , \mathcal{H}_{i})$, $\psi_{ij}: \mathcal{B}(\mathcal{H}_{1})\mapsto \mathcal{B}(\mathcal{H}_{j} , \mathcal{H}_{i})$ two families of linear maps.  Since $\Phi$ is completely positive, $\phi_{ii}$ and $\psi_{ii}$ are completely positive too for $i=0,1$. The unitality condition $\Phi(\mathbb{I}_{\mathcal H})=\mathbb{I}_{\mathcal H}$ implies 
	\begin{align}
		&\phi_{ij}( \mathbb{I}_{\mathcal{H}_0}) + \psi_{ij} (\mathbb{I}_{\mathcal{H}_1})= 0, \quad i , j=0,1,\; i \neq j, \\
		&\phi_{ii}(\mathbb{I}_{\mathcal{H}_0}) + \psi_{ii}(\mathbb{I}_{\mathcal{H}_1})=\mathbb{I}_{\mathcal{H}_{i}}, \quad i=0,1 ,
	\end{align}
	Since $\phi_{00}$ is a UCP map, we have $\phi_{00} (\mathbb{I}_{\mathcal{H}_0})=\mathbb{I}_{\mathcal{H}_{0}}$ implying that $\psi_{00}(\mathbb{I}_{\mathcal{H}_1})=0$. Since it is a positive map, we have $\psi_{00}=0$ by a corollary of the Russo-Dye theorem~\cite[Corollary 2.3.8]{bhatia2009positive}, and from the positivity of $\Phi$ we have $\psi_{01}=\psi_{10}=0$ by~\cite[Proposition 1.3.2]{bhatia2009positive}.
	
	Therefore, we can write the action of $\Phi$ on an element $X=X_{00}\oplus X_{11}$ as
	\begin{equation}
		X=
		\begin{pmatrix}
			X_{00}& 0\\
			0 & X_{11}
		\end{pmatrix}\mapsto 
		\Phi(X)=\begin{pmatrix}
			{\phi}_{00}(X_{00}) & \phi_{01}(X_{00}) \\
			\phi_{10}(X_{00}) & \phi_{11}(X_{00}) + \psi_{11}(X_{11})
		\end{pmatrix},
	\end{equation}
	where $\phi_{00}$ is a faithful UCP map, $\phi_{11}, \psi_{11}$ are completely positive maps satisfying
	\begin{equation}
		\phi_{11}(\mathbb{I}_{\mathcal{H}_0}) + \psi_{11}(\mathbb{I}_{\mathcal{H}_1})=\mathbb{I}_{\mathcal{H}_{1}},
	\end{equation}
	and $\phi_{ij}$ are such that $\phi_{ij}( \mathbb{I}_{\mathcal{H}_0})= 0$.
	Note that, at variance with $\Phi^\dagger$, $0 \oplus \mathcal{B}(\mathcal{H}_{1})$ is invariant under the action of $\Phi$, while $\mathcal{B}(\mathcal{H}_0) \oplus 0$ is not.
	
	Take now $X=X_{00} \oplus 0$ with $X_{00}\in \mathrm{Attr}(\phi_{00})$. Then $\phi_{10}(X_{00})=\phi_{01}(X_{00})=0$ because of the operator Schwarz inequality~\eqref{op_Schw_ineq}, by repeating the same computation allowing us to derive equation~\eqref{eq:condition_off_channel}. So, given $X\in \mathrm{Attr}(\Phi)$, the action of $\Phi$ reads
	\begin{equation}
		\Phi(X)= \begin{pmatrix}
			\phi_{00}(X_{00}) & \\
			& \phi_{11}(X_{00}) + \psi_{11}(X_{11})
		\end{pmatrix},
	\end{equation}
	with $X_{11}=\mathcal P_{11}(X_{00})$. The eigenvalue equation $\Phi(X)=\lambda X$ with $\lambda \in \mathrm{spect}_{\mathrm{P}}(\Phi)$ can be written explicitly as 
	\begin{align}
		&\phi_{00}(X_{00})=\lambda X_{00},\label{eigen_co1}\\
		&\phi_{11}(X_{00}) + \psi_{11}(X_{11})=\lambda X_{11}. \label{eigen_co2}
	\end{align}
	We can rewrite the latter as
	\begin{align}
		(\lambda - \psi_{11})^{-1}\phi_{11}(X_{00})= X_{11},
		\label{eigen_per_prob_2}
	\end{align}
	a result obtained in~\cite[Proposition 3]{albert2019asymptotics}. 
	
	Observe that $\lambda - \psi_{11}$ is invertible, or equivalently $\lambda \not\in \mathrm{spect}(\psi_{11})$. Indeed, if this were not true, viz. if there existed $0 \neq Y_{11}\in \mathcal{B}(\mathcal{H}_{1})$ so that 
	\begin{equation}
		\psi_{11}(Y_{11})=\lambda Y_{11},
	\end{equation}
	then
	\begin{equation}
		\Phi(0 \oplus Y_{11})=\lambda( 0 \oplus Y_{11}),
	\end{equation}
	and $0 \oplus Y_{11} \in \mathrm{Attr}(\Phi)$, in stark contrast with equation~\eqref{struc_attr_Heis}. In particular, by choosing $\lambda = 1$, we find that $\psi_{11}$ cannot be unital or, equivalently, $\phi_{11} \neq 0$.
	
	Incidentally, notice also that the maps $\phi_{11}$, $\psi_{11}$ and $\phi_{00}$ are related to the maps $\mathcal{P}_{ii}$ describing the action of the peripheral projection $\mathcal{P}_{\mathrm{P}}$ of $\Phi$ via
	\begin{equation}
		\label{comm_co}
		\phi_{11} \circ \mathcal{P}_{00}+\psi_{11} \circ \mathcal{P}_{11} = \mathcal{P}_{11} \circ \phi_{00},
	\end{equation}
	arising from the commutation relation $[\mathcal{P}_{\mathrm{P}} , \Phi]=0$.  Consistently, the left-hand side of~\eqref{eigen_co2}  becomes as a consequence of~\eqref{comm_co}
	\begin{equation}
		\phi_{11}(X_{00}) + \psi_{11}(X_{11})= \mathcal{P}_{11} \circ \phi_{00}(X_{00}) = \lambda \mathcal{P}_{11} (X_{00}) ,
	\end{equation}
	which is in line with~\eqref{as-struc_Heis}.
	
	To sum up, one has the following result.
	
	\begin{thm}
		\label{per_eigen_th}
		Let $\Phi$ be a UCP map with peripheral projection $\mathcal{P}_{\mathrm{P}}$ and let $\mathcal{H}_{0}=\mathrm{supp}(\mathcal{P}_{\mathrm{P}}^\dagger(\mathbb{I}))$ be the maximal support of the elements of its attractor subspace $\mathrm{Attr}(\Phi)$ and $\mathcal{H}_{1}=\mathcal{H}_{0}^\perp$. Then the action of $\Phi$ on a block diagonal operator $X= X_{00} \oplus X_{11} \in \mathcal{B}(\mathcal{H}_{0}) \oplus \mathcal{B}(\mathcal{H}_{1})$ reads
		\begin{equation}
			\label{diag_action_Phi_2}
			\Phi(X)=\begin{pmatrix}
				\phi_{00}(X_{00}) & \phi_{01}(X_{00}) \\
				\phi_{10}(X_{00}) & \phi_{11}(X_{00}) + \psi_{11}(X_{11})
			\end{pmatrix},
		\end{equation}
		where $\phi_{00}$ is a faithful UCP map, $\phi_{11}, \psi_{11}$ are completely positive maps satisfying
		\begin{equation}
			\phi_{11}(\mathbb{I}_{\mathcal{H}_0}) + \psi_{11}(\mathbb{I}_{\mathcal{H}_1})=\mathbb{I}_{\mathcal{H}_{1}},
		\end{equation}
		and $\phi_{ij}$ $i \neq j$ are linear maps with $\phi_{ij}(\mathbb{I}_{\mathcal{H}_0})=0$. 
		
		Consequently, the eigenoperator $X^{(\lambda)}$ of $\Phi$ corresponding to the peripheral eigenvalue $\lambda$, 
		\begin{equation}
			\Phi(X^{(\lambda)})=\lambda X^{(\lambda)} , \quad \lambda \in \mathrm{spect}_{\mathrm{P}}(\Phi),
		\end{equation}
		has the form
		\begin{equation}
			\label{Xlambda_struc}
			X^{(\lambda)} = \begin{pmatrix}
				X_{00}^{(\lambda)} & \\
				& \mathcal{P}_{11}(X_{00}^{(\lambda)})
			\end{pmatrix}, \quad \text{with} \quad \phi_{00}(X_{00}^{(\lambda)}) = \lambda X_{00}^{(\lambda)}.
		\end{equation}
	\end{thm} 
	
	Observe that Theorem~\ref{per_eigen_th} still holds for Schwarz maps, with the only difference that the maps $\phi_{00}$, $\phi_{11}$ and $\psi_{11}$ would  not be completely positive in general, but they just satisfy the operator Schwarz inequality~\eqref{op_Schw_ineq}. In particular, in order to derive~\eqref{Xlambda_struc}, i.e.\ the structure of the peripheral eigensystem of $\Phi$, we do not need complete positivity and, in particular, the Kraus representation of $\Phi$ used to obtain it in~\cite{albert2019asymptotics}, but only the operator Schwarz inequality~\eqref{op_Schw_ineq}.
	
	\section*{Acknowledgements}
    We are grateful for various comments from the Referee.\
	DA also thanks Federico Girotti for pointing out Reference~\cite{gartner2012coherent}.  AK acknowledges the support by the (Polish) National Science Center through the SONATA BIS Grant No. 2019/34/E/ST2/00369. AK also acknowledges the support from the ``Young Scientist" project at the Center for Theoretical Physics, Polish Academy of Sciences. We acknowledge support from INFN through the project ``QUANTUM'', from the Italian National Group of Mathematical Physics (GNFM-INdAM), and from the Italian funding within the ``Budget MUR - Dipartimenti di Eccellenza 2023--2027''  - Quantum Sensing and Modelling for One-Health (QuaSiModO).\ PF acknowledges support from Italian PNRR MUR project CN00000013 -``Italian National Centre on HPC, Big Data and Quantum Computing''. DA acknowledges support from PNRR MUR project PE0000023-NQSTI.

\endpaper
\end{document}